\documentclass[conference]{IEEEtran}
\usepackage[a4paper,left=2cm,right=2cm,top=2.5cm,bottom=2.5cm]{geometry}
\usepackage[utf8]{inputenc}
\usepackage{nopageno}
\usepackage{graphicx,float,caption,subcaption}
\usepackage{minted}
\usepackage[dvipsnames]{xcolor}
\usepackage{authblk}          
\usepackage{hyperref}
\hypersetup{
    colorlinks=true,
    linkcolor=blue,
    filecolor=magenta,      
    urlcolor=blue,
    citecolor=PineGreen,
    pdftitle={Bootstrapping Thermality},
    pdfpagemode=FullScreen,
    }

\usepackage{amsmath,amssymb,amsthm,bm,graphicx,float,array,multirow,multicol,rotfloat,hyperref,cleveref,enumerate,geometry,mathdots,adjustbox,booktabs,mathtools,tikz,tikz-cd,pdflscape,csquotes,lscape,rotating,empheq,mathrsfs,subcaption}
\usetikzlibrary{calc, arrows.meta, shapes.geometric, patterns, positioning, matrix, backgrounds, decorations.pathreplacing, shadows}
\usepackage{shadethm}
\usepackage[para]{threeparttable}
\usepackage[all]{xy}
\usepackage[normalem]{ulem}
\usepackage[numbers,sort&compress]{natbib}
\usepackage{physics}
\usepackage{algorithm}
\usepackage{algorithmic}
\newtheorem{thm}{Theorem}

\newtheorem{lem}[thm]{Lemma}

\newtheorem{dfn}[thm]{Definition}

\title{Thermal reconstruction of chaotic quantum many-body systems}
\author[1]{Shozab Qasim}
\author[2,3]{Jason Pollack}
\affil[1]{Department of Physics, Dahlem Center for Complex Quantum Systems, Freie Universität Berlin}
\affil[2]{Department of Electrical Engineering and Computer Science, Syracuse University}
\affil[3]{Institute for Quantum \& Information Sciences, Syracuse University}
\date{}

\begin{document}

\maketitle

\begin{abstract}
Thermal states are thermal with respect to a fixed Hamiltonian. How much information about this Hamiltonian can we ``bootstrap'' from the subsystems of a thermal state? We attack the problem by positioning it as a subspecies of the quantum marginal problem. In states that obey the quantum Markov property, the Petz recovery map captures the knowledge of the larger system inherent in a subsystem. We use the conditional mutual information to check the goodness of Petz recovery, analytically in a random-matrix-theory-inspired hopping model and numerically in an Ising-like spin chain model. We observe different behavior in chaotic versus integrable phases of the model: in the chaotic phase, the reconstruction works well at both very low and very high temperatures, with some intermediate critical temperature at which reconstruction works worst, whereas in the integrable phase reconstruction breaks down at low temperatures. 
\end{abstract}
\tableofcontents 

\section{Introduction}
\begin{quote}
Can the role of the wavefunction in quantum mechanics be supplanted by reduced density matrices? - A. J. Coleman \cite{Coleman}
\end{quote}
A central difficulty in quantum many-body physics is that the underlying physics depends on the full many-body state. The quantum marginal problem \cite{marginalphysical, Klyachko_2006} asks a basic but subtle question: when are reduced density matrices on small subsystems consistent with some global quantum state? A practical use case for this problem is that instead of computing physical observables with the full density matrix, one can instead compute physical observables with the marginals. 

The grand goal is to construct a numerical algorithm where it is more efficient to compute physical quantities with the marginals instead of the full global wavefunction (for applications to quantum chemistry, see \cite{Mazziotti_2012, Mazziotti_2023}). As a concrete example, consider a lattice system consisting of $n$ sites with local Hilbert-space dimension $d$. On some subsets $S_i$ of the lattice sites, we specify reduced states $\rho_i$. We wish to know whether $\rho_i = \Tr_{\backslash S_i} \rho$ for some global $\rho$. A solution to this problem would also solve all finite-dimensional few-body ground state problems: for a two-body Hamiltonian, $H = \sum_{i,j =1}^n h_{ij}$, one need only compute 
\begin{equation}
    \min_{\rho} \Tr H \rho = \min_{\rho} \sum_{i,j} \Tr h_{ij} \rho = \min_{ \{\rho_{i,j} \text{comp.}\} } \sum_{i,j} \Tr h_{i,j} \rho_{i,j}.
\end{equation}
The far left hand side of the equation optimizes over $O(d^n)$ variables whereas the far right hand side optimizes over $O(n^2 d^4)$ variables, which is an exponential improvement. The optimization is over a convex set of compatible $\rho_{i,j}$. The computational complexity of this 2-site reduced density matrix is dominated by the complexity of deciding compatibility of $\rho_{i,j}$'s. Finding such two-body ground states in general is NP-hard and in some cases even QMA hard. So there is no efficient quantum or classical algorithm for the most general version of the two-body marginal problem. However, for physically relevant instances, approximations have been successful \cite{Mazziotti_2012, Mazziotti_2023}. The general question of the quantum marginal problem has attracted a respectable amount of attention in the quantum information community \cite{klyachko2004quantummarginalproblemrepresentations, Klyachko_2006, Lopes2015, eisert2023notelowerboundsvariational, Eisert_2008} with subclasses of solutions leveraging entanglement polytopes \cite{Walter_2013}. 

A practical example where a solution of the marginal problem is extremely simple is mean-field theory - which is known to be a good description if the entanglement is low \cite{trimborn}. More precisely, if there are local density matrices $\rho_1, \dots, \rho_m$, each supported on a set of at most $k$ over a system of $n$ qudits, the global state $\tau$ that is compatible with the marginals, i.e., the mean-field solution, is simply $\tau = \rho_1 \otimes \dots \otimes \rho_m$. Mean-field theory, however, captures a very limited entanglement structure that fails for many interesting highly-entangled phases of matter like topological phases. Alternatives to mean-field theory include tensor networks \cite{TNreview} with very successful algorithms \cite{Ba_uls_2023} for capturing properties of phases of matter in 1D like DMRG. The situation is much harder in 2D, but progress has been made due to recent breakthroughs in characterizing entanglement in 2D \cite{kim1, kim2}. 

These works focus on ground-state physics. A much more challenging problem is determining compatibility of marginals with global thermal states. This is intimately tied to the physics of thermalization, which is often stated to be the hardest problem in quantum many-body physics as it requires access to the entire energy spectrum \cite{thermalization}. Studies of thermalization from the physics perspective are thus mainly numerical in nature. Analytic statements are possible only if one resorts to using random matrices. In this paper, we are concerned with combating the marginal problem in the chaotic quantum many-body systems for which the Eigenstate Thermalization Hypothesis \cite{Srednicki:1994mfb, Deutsch:1991msp,Jensen:1984gu} is believed to hold true for operators of low complexity. Our goal is to numerically characterize the entanglement of the thermal state of a chaotic many-body system as a function of temperature. We numerically demonstrate that a Markovianity assumption is satisfied, which means that the conditional mutual information is low. We can then recover the global thermal state from the marginals using an information-theoretic measure called the Petz recovery map. Additionally, we demonstrate how the Markovianity assumption is satisfied in local random band matrices. Finally, we prove bounds on the closeness of dynamics between the recovered Petz state and the original state.  

We are motivated to investigate these questions in the hope that just as the marginal problem has been successful for understanding ground states, it may also shed some light on chaotic systems, and that marginals of thermal states of chaotic systems could be used to compute thermal expectation values and perhaps also capture dynamics. The end goal would be an algorithm that would numerically compute thermal density matrices. This is a rather lofty goal, and we only attempt to characterize the following problem: how well does the Petz map reconstruct thermal states of chaotic many-body systems?  

Previous analysis of these questions restricted to ground states are \cite{Hu_2024, vardhan2023Petzrecoverysubsystemsconformal, jia2020Petzreconstructionrandomtensor}, results on thermal states of stabilizers can be found in \cite{ temme2015faststabilizerhamiltoniansthermalize} and some work on thermal states is in \cite{Alhambra_2017}.

Spectra of chaotic quantum many-body Hamiltonians are captured by random matrices by virtue of the Wigner surmise \cite{BGS1984}. Thermalization can be studied either analytically by leveraging random matrices \cite{qasim2025emergentstatisticalmechanicsholographic, Haferkamp_2021, Harrow:2022znr, Bertoni_2025, wang2025eigenstatethermalizationhypothesisrandom, jafferis2023jtgravitymattergeneralized, jafferis2023matrixmodelseigenstatethermalization} or numerically for quantum many-body systems \cite{cáceres2024genericetheigenstatethermalization}. In this work, we make use of both methods: we compute the Petz recovery for random band matrices and compute the Petz recovery fidelity for realistic spin chains. For more background, see \cite{Abanin_2019, Gogolin_2016, Eisert_2015} and references therein. 

\subsection{Summary of results}
\begin{itemize}
    \item We provide analytical bounds on the closeness of dynamics and expectation values of observables under the Petz map. 
    \item We numerically reconstruct thermal states of a spin chain model at various temperatures. We see a key difference between integrable and chaotic systems: in the chaotic phase, the reconstruction works well at both very low and very high temperatures, with some intermediate critical temperature at which reconstruction works worst, whereas in the integrable phase reconstruction breaks down at low temperatures. 
    \item We reconstruct thermal states of random band matrices analytically in a perturbative limit.     
\end{itemize}

The remainder of this paper is organized as follows. In Section \ref{sec:prelim} we review the marginal problem, Petz recovery, random matrix theory, and quantum chaos. Section \ref{sec:results} presents our analytic and numerical results. We discuss and conclude in Section \ref{sec:discussion}. A brief appendix contains details of derivations. 

\section{Preliminaries}\label{sec:prelim}
We will explain the marginal problem more formally and explain the Petz recovery map. Finally, we will explain how random matrix theory applies in superselection sectors of spin chains. Readers familiar with these descriptions may skip to the next section. 

\subsection{The marginal problem, Markovianity and Petz recovery}

Given density operators for subsystems of a
multipartite quantum system are they compatible to one common total density operator? This is known as the quantum marginal problem:

\begin{dfn} \textbf{Marginal Problem} \cite{schilling2014quantummarginalproblem}
    For a given family $\mathcal{K}$ of subsystems $\mathcal{I}$ of $\mathcal{J}$ the quantum marginal problem $\mathcal{M}_\mathcal{K}$ is the problem of determining and
    describing the set $\sigma_\mathcal{K}$ of tuples $(\rho_\mathcal{I})_{\mathcal{I}} \in \mathcal{K}$ of compatible marginals. Compatible here means that there exists a density operator $\rho_\mathcal{J}$ for the total system such that $\forall{I} \in \mathcal{K}$
    \begin{equation}
        \rho_\mathcal{I} = \Tr_{\mathcal{J} \backslash \mathcal{I}}[\rho_\mathcal{J}]. 
    \end{equation}
\end{dfn}

A state on a tripartite quantum system $A\otimes B \otimes C$ forms a \textbf{Markov chain} if it can be reconstructed
from its marginal on $A\otimes B$ by a quantum operation from $B$ to $B\otimes C$. We show that the quantum
conditional mutual information $I(A : C|B)$ (Where the conditional mutual information is defined as $I(A:C \mid B) = S(\rho_{AB}) + S(\rho_{BC}) - S(\rho_{ABC}) - S(\rho_B)$) of an arbitrary state is an upper bound on its distance to the closest reconstructed state. It thus quantifies how well the Markov chain property is approximated. This is summarized by the following theorem:
\begin{thm} \cite{Fawzi_2015}
    For any density operator $\rho_{ABC}$ on $A\otimes B \otimes C$, where $A, B,$ and $C$ are separable Hilbert spaces, there exists a trace-preserving completely positive map $\mathcal{R}_{B\xrightarrow{} BC}$ from the space of operators on $B$ to the space of operators on $B \otimes C$ such that 
    \begin{equation}
       2^{-\frac{1}{2}I(A: C|B)_\rho} \leq F(\rho_{ABC}, (\mathcal{I}_A \otimes \mathcal{R}^{(\lambda)}_{B \xrightarrow{} BC}(\cdot))\rho_{AB}). 
    \end{equation}
    Furthermore if $A, B$ and $C$ are finite dimensional then $\mathcal{R}^{(\lambda)}_{B \xrightarrow{} BC}(\cdot)$ has the form 
\begin{equation}
    \mathcal{R}^{(\lambda)}_{B \xrightarrow{} BC}(\cdot) = \rho_{BC}^{\frac{1}{2}-\frac{i \lambda}{2}} \rho_{B}^{-\frac{1}{2}+\frac{i \lambda}{2}}(\cdot)\rho_{B}^{-\frac{1}{2}-\frac{i \lambda}{2}} \rho_{BC}^{\frac{1}{2}+\frac{i \lambda}{2}}.
\end{equation}
This most general form is called the Petz recovery map, we will study this in the case where $\lambda=0$.
\end{thm}
The setup is depicted in Figure \ref{fig:Petzrecovery}.

\begin{figure*}[h]
    \centering
\begin{tikzpicture}[>=Latex, every node/.style={font=\large}]

\begin{scope}
  \node at (1.5,1.4) {$\rho_{ABC}$};
  \fill[orange!25] (0,0) rectangle (3,1);
  \draw (0,0) rectangle (3,1);
  \draw (1,0) -- (1,1);
  \draw (2,0) -- (2,1);
  \node at (0.5,0.5) {$A$};
  \node at (1.5,0.5) {$B$};
  \node at (2.5,0.5) {$C$};
\end{scope}

\draw[->,thick]
      (3.2,0.5) -- ++(1.4,0)
      node[midway,above=4pt] {$\mathrm{Tr}_{C}$};

\begin{scope}[xshift=4.8cm]
  \node at (1,1.4) {$\rho_{AB}$};
  \fill[orange!25] (0,0) rectangle (2,1);
  \draw (0,0) rectangle (2,1);
  \draw (1,0) -- (1,1);
  \node at (0.5,0.5) {$A$};
  \node at (1.5,0.5) {$B$};
\end{scope}

\draw[->,thick]
      (6.9,0.5) -- ++(1.7,0)
      node[midway,above=4pt] {$\mathcal{R}_{B\!\to\!BC}$};

\begin{scope}[xshift=8.8cm]
  \node at (1.5,1.4) {$\tilde{\rho}_{ABC}$};
  \fill[green!25] (0,0) rectangle (3,1);
  \draw (0,0) rectangle (3,1);
  \draw (1,0) -- (1,1);
  \draw (2,0) -- (2,1);
  \node at (0.5,0.5) {$A$};
  \node at (1.5,0.5) {$B$};
  \node at (2.5,0.5) {$C$};
\end{scope}

\end{tikzpicture}
    \caption{ We trace out the subsystem $C$ from a state $\rho_{ABC}$ , and then attempt to recover the original state with a channel $\mathcal{R}_{B \xrightarrow{} BC}$ that acts non-trivially only on $B$.}
    \label{fig:Petzrecovery}
\end{figure*}
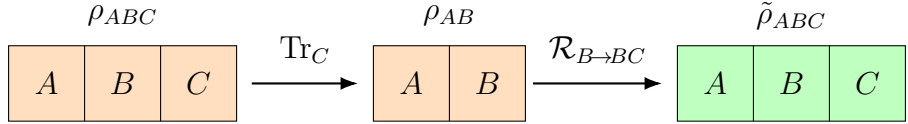 

\subsection{Random matrix theory, quantum chaos, thermalization and superselection sectors}

Random matrix theory (RMT) applies within each superselection sector of a chaotic system. This means that, for a Hilbert space split in the manner 
\begin{equation}
    \mathcal{H} = \mathcal{H}_1 \oplus \mathcal{H}_2 \oplus \dots \oplus \mathcal{H}_n,
\end{equation}
where time evolution given by the Hamiltonian keeps states within the sectors $\mathcal{H_i},$ RMT behaviour (as seen from spectral statistics) is replicated within each superselection sector within the middle of the spectrum. 

\begin{figure}[ht]
\centering
\includegraphics[width=0.45\textwidth]{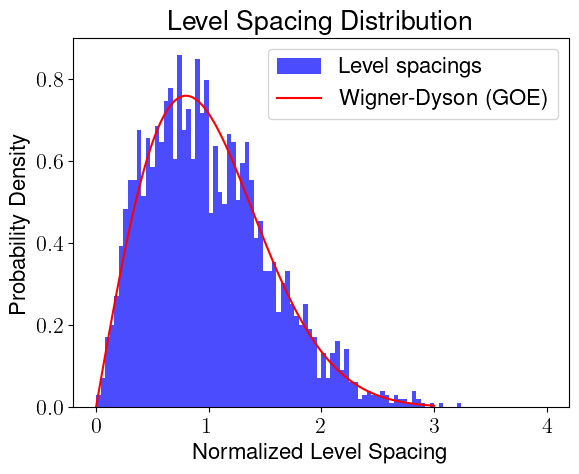}
\caption{Level spacing statistics for the Hamiltonian eq \ref{eq:ham} for the parameters $\alpha=1.0, J_x=1.05, h_z=-0.5$ within a superselection sector (parity block) of $14$ spins.}
\label{fig:WD}
\end{figure}
Thermalization is the empirical statement that energy eigenstates typically appear thermal when probed. The Eigenstate Thermalization Hypothesis (ETH) is a mathematical ansatz for the matrix elements of simple operators $\mathcal{O}$ in the eigenbasis of a chaotic Hamiltonian which captures the behavior of the energy eigenstates and displays how they are more than just random vectors. It reads 
\begin{equation}
    \bra{E_m}\mathcal{O}\ket{E_n} = \overline{O}(E)\delta_{mn} + e^{-S(\overline{E})/2}f_O(E,\omega)R_{mn},
\end{equation}
where $\overline{O}(E)$ the microcanical expectation value $f_O(E,\omega)$ a smooth spectral function of the averaged energy $E = (E_m + E_n)/2$ and the energy difference $\omega = E_m -E_n$ satisfying $f_O(E,-\omega)=f_{O}(E,\omega)$ for real Hamiltonians $R_{mn}$ a random variable with zero mean, unit variance satisfying $R_{mn}=R_{nm}$. The thermodynamic entropy (also called the microcanonical entropy) $S(E)$ is defined as the logarithm of the coarse-grained density of states i.e. the number of eigenstates of energy $E$ is given by $e^{S(E)}$.

Our model chaotic system will be an Ising-type Hamiltonian:
\begin{equation}\label{eq:ham}
H = \alpha \sum_i^{L-1} \sigma_i^z\sigma_{i+1}^z + h_z \sum_i^{L} \sigma_i^z + J_x \sum_i^L \sigma_i^x, 
\end{equation}
using the parameters $\alpha=1.0, J_x=1.05, h_z=-0.5$. For open boundary conditions, the Hamiltonian has a parity symmetry with two superselection sectors; we'll elaborate on this point in a moment. If we pick one of the superselection sectors, we find using the Quspin package \cite{10.21468/SciPostPhys.2.1.003} that the level spacing statistics follows Wigner-Dyson statistics, as shown in Figure \ref{fig:WD}. 

For an $L$-site spin chain, the parity operator $\mathcal{P}$ acts on a Pauli spin as
\begin{equation}
    \mathcal{P} \sigma_i^\alpha \mathcal{P}^\dagger = \sigma_{L+1-i}^\alpha,
\end{equation}
where $\alpha \in {X,Y,Z}$. In other words it takes site $i$ to site $L+1-i$. The wave function can accordingly decomposed into a parity-even and a parity-odd piece: $\ket{\psi}=\ket{\psi_{+1}}+\ket{\psi_{-1}}$, where $\mathcal{P}\ket{\psi_{+1}}=\ket{\psi_{+1}}$ and $\mathcal{P}\ket{\psi_{-1}}=-\ket{\psi_{-1}}$. Concretely, each computational basis state is acted on as
\begin{equation}
    \mathcal{P} \ket{s_1,s_2, \dots, s_L} = \ket{s_L,s_{L-1}, \dots, s_1}.
\end{equation}




\section{Results}\label{sec:results}
In this section we will present our results. We start by proving several model-independent bounds on how much the expectation values of observables differ between the reconstructed thermal state and the original thermal state. We then test these bounds, and thermal reconstruction in general, using numerics on a spin chain with both integrable and chaotic phases. We complement these results with a perturbative analysis of the fidelity between the thermal and reconstructed states for a simple chaotic random band matrix model, although we defer the details to the Appendix. 

\subsection{Closeness of dynamics and expectation values of observables}
To what extent will the dynamical properties of the recovered state match with those of the original state? We present a series of lemmas that bound the differences between the recovered and original states. 
\begin{lem}
    Given $ 2^{-\frac{1}{2}I(A: C|B)_\rho} \leq F(\rho_{ABC}, \widetilde{\rho}_{ABC})$, where $\widetilde{\rho}_{ABC}$ is the reconstructed state via the Petz map, $ 2^{-\frac{1}{2}I(A: C|B)_\rho} \leq F(\mathcal{N}(\rho_{ABC}), \mathcal{N}(\widetilde{\rho}_{ABC}))$ for some quantum channel $\mathcal{N}$. Additionally, $\norm{\mathcal{N}(\rho_{ABC})-\mathcal{N}(\widetilde{\rho}_{ABC})}_1 \leq \epsilon$ where $\epsilon = \norm{\rho - \widetilde{\rho}}_1 \leq \sqrt{4(1-2^{-\frac{1}{2}I(A: C|B)_\rho})}$.
\end{lem}
\begin{proof}
    This is due to a straightforward application of the data processing inequality where $F(\Lambda(\rho),\Lambda(\sigma)) \geq F(\rho,\sigma)$ for CPTP maps $\Lambda$ and density matrices $\rho$ and $\sigma$. The only thing left to do is bound $\epsilon$ for which we use Fuchs-van de Graff. $1- \frac{1}{2}\norm{\rho-\sigma}_1 \leq F(\rho,\sigma) \leq \sqrt{1-\frac{1}{4}\norm{\rho - \sigma}_1^2}$ - we will use the right hand side of the inequality from which it follows that $\epsilon \leq \sqrt{4(1-2^{-\frac{1}{2}I(A: C|B)_\rho})}$.
\end{proof}
The physical content of this lemma is that evolution of the density matrix $\rho$ and the Petz reconstructed density matrix $\widetilde{\rho}$ under a fixed Hamiltonian will be suppressed exponentially in the CMI of the density matrix $\rho$. 

\begin{lem}
    Given $ 2^{-\frac{1}{2}I(A: C|B)_\rho} \leq F(\rho_{ABC}, \widetilde{\rho}_{ABC})$, where $\widetilde{\rho}_{ABC}$ is the reconstructed state via the Petz map, and two quantum channels $\mathcal{N}_1, \mathcal{N}_2: \mathcal{B}(\mathcal{H}) \xrightarrow{} \mathcal{B}(\mathcal{K})$ with a diamond norm distance $\delta= \norm{\mathcal{N}_1 - \mathcal{N}_2}_\diamond$, then $\norm{\mathcal{N}_1(\rho) - \mathcal{N}_2(\widetilde{\rho})}_1 \leq \epsilon + \delta$ where $\epsilon = \norm{\rho - \widetilde{\rho}}_1 \leq \sqrt{4(1-2^{-\frac{1}{2}I(A: C|B)_\rho})}$.
\end{lem}
\begin{proof}
    Recall $\norm{\mathcal{N}}_\diamond = \sup_{\norm{\psi}_1\leq 1} \norm{(\mathcal{N} \otimes I)(\psi)}_1$. Then $\mathcal{N}_1(\rho) - \mathcal{N}_2(\sigma) = \mathcal{N}_1(\rho-\sigma) + (\mathcal{N}_1-\mathcal{N}_2)(\sigma)$ due to linearity. Following that, $\norm{\mathcal{N}_1(\rho)-\mathcal{N}_2(\sigma)}_1 \leq \norm{\mathcal{N}_1(\rho - \sigma)}_1 + \norm{(\mathcal{N}_1 - \mathcal{N}_2)(\sigma)}_1$ from which we obtain $\norm{\mathcal{N}_1(\rho)-\mathcal{N}_2(\widetilde{\rho})}_1 \leq \epsilon + \delta$. The only thing left to do is bound $\epsilon$ for which we use Fuchs-van de Graff: $1- \frac{1}{2}\norm{\rho-\sigma}_1 \leq F(\rho,\sigma) \leq \sqrt{1-\frac{1}{4}\norm{\rho - \sigma}_1^2}$ - we will use the right hand side of the inequality from which it follows that $\epsilon \leq \sqrt{4(1-2^{-\frac{1}{2}I(A: C|B)_\rho})}$.
\end{proof}
The physical content of this lemma is that evolution of the density matrix $\rho$ and the Petz reconstructed density matrix $\widetilde{\rho}$ under a two different Hamiltonians will differ in terms of terms that depend on the CMI of the original state and the norm difference between the two Hamiltonians. 

We would like to compare the difference in measurement outcomes for the state $\rho$ and the Petz reconstructed state $\widetilde{\rho}$. A simple way to do this is to bound the trace norm difference, note that the trace norm of a matrix is the sum of the singular values which are the eigenvalues of the product of the matrix and its complex conjugate. As a precursor, consider how the elements would deviate in the ETH between two energy eigenbasis $\{\ket{E_i}\}$ and $\{\ket{E_i'}\}$:
\begin{align}
    &\bra{E_i} \mathcal{O} \ket{E_i} - \bra{E_i'} \mathcal{O} \ket{E_i'} = [\overline{O}(E)-\overline{O}(E')]\delta_{mn}\nonumber\\
    &+ [e^{-S(\overline{E})}f_{O}(E,\omega) - e^{-S(\overline{E'})}f_{O}(E',\omega')]R_{ij}.
\end{align}
We attempt to quantify this difference in terms of difference in expectation values of operators in with the reconstructed state and the original state using the following lemmas:
\begin{lem}
    Given $ 2^{-\frac{1}{2}I(A: C|B)_\rho} \leq F(\rho_{ABC}, \widetilde{\rho}_{ABC})$ where $\widetilde{\rho}_{ABC}$ is the reconstructed state via the Petz map, $\norm{\mathcal{O}(\rho - \widetilde{\rho})}_1 \leq \norm{\mathcal{O}}_\infty \sqrt{4(1-2^{-\frac{1}{2}I(A: C|B)_\rho})}$ where
    $\epsilon = \norm{\rho - \widetilde{\rho}}_1 \leq \sqrt{4(1-2^{-\frac{1}{2}I(A: C|B)_\rho})}$.
\end{lem}
\begin{proof}
    This results from a straightforward application of Holders inequality by which $\norm{\mathcal{O}(\rho - \widetilde{\rho})}_1 \leq \norm{\mathcal{O}}_p \norm{\rho - \widetilde{\rho}}_q$ where $\frac{1}{p}+\frac{1}{q}=1$. From this its clear that the upper bound is $\epsilon \norm{\mathcal{O}}_\infty $.
\end{proof}

\begin{lem}
Given thermal states $\rho=e^{-\beta H}$ and the Petz reconstructed state $\widetilde{\rho}=e^{-\beta \widetilde{H}}$ which satisfy $ 2^{-\frac{1}{2}I(A: C|B)_\rho} \leq F(\rho, \widetilde{\rho})$, a high-temperature expansion (small values of $\beta$) yields a trace norm difference between the Hamiltonians given by $\norm{\widetilde{H}-H}_1\leq \frac{\delta}{\beta}$ where $\delta=\sqrt{1-\frac{1}{4}\norm{\rho-\widetilde{\rho}}_1^2}$.
\end{lem}

\begin{proof}
We start with $2^{-I} \leq F(e^{-\beta H}, e^{-\beta \widetilde{H}})$ which reduces to $\norm{e^{-\beta H} - e^{-\beta \widetilde{H}}}_1 \leq \delta$. Performing a first-order Taylor expansion, we obtain $\norm{(1-\beta H)-(1-\beta \widetilde{H})} \leq \delta$ which simplifies to $\norm{\widetilde{H} - H }_1 \leq \frac{\delta}{\beta}$ and is a bound valid for small $\beta$.
\end{proof}
 
\subsection{Numerics}
Statements about thermalization can be proven analytically in systems involving randomness \cite{qasim2025emergentstatisticalmechanicsholographic, Haferkamp_2021} , an example of which we will investigate in Subsection \ref{thermalrecons}. However, an analysis of thermalization in physical systems requires numerical investigation. We thus resort to the usage of exact diagonalization methods in order to investigate the quality of Petz recovery for thermal states of physical systems.

\begin{figure*}[t!]
    \centering
    \begin{subfigure}[b]{0.4\linewidth}
        \includegraphics[width=\linewidth]{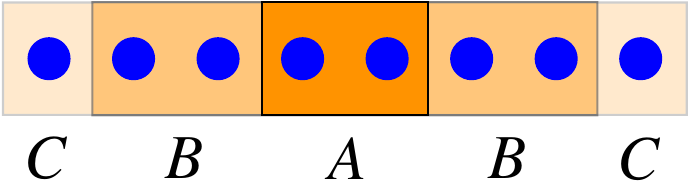}
        \caption{Permuted configuration.}
    \end{subfigure}
    \hfill
    \begin{subfigure}[b]{0.4\linewidth}
        \includegraphics[width=\linewidth]{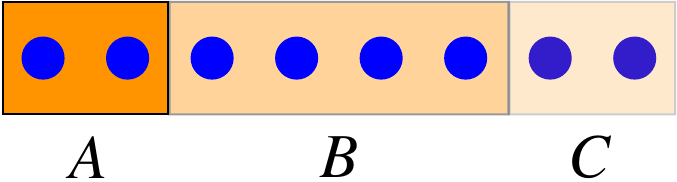}
        \caption{Non-permuted configuration.}
    \end{subfigure}
    \caption{Comparison between permuted and non-permuted configurations.}
    \label{fig:configs}
\end{figure*}

\begin{figure*}[h!]
\centering
\begin{subfigure}[t]{0.45\textwidth}
    \centering
    \includegraphics[width=\textwidth]{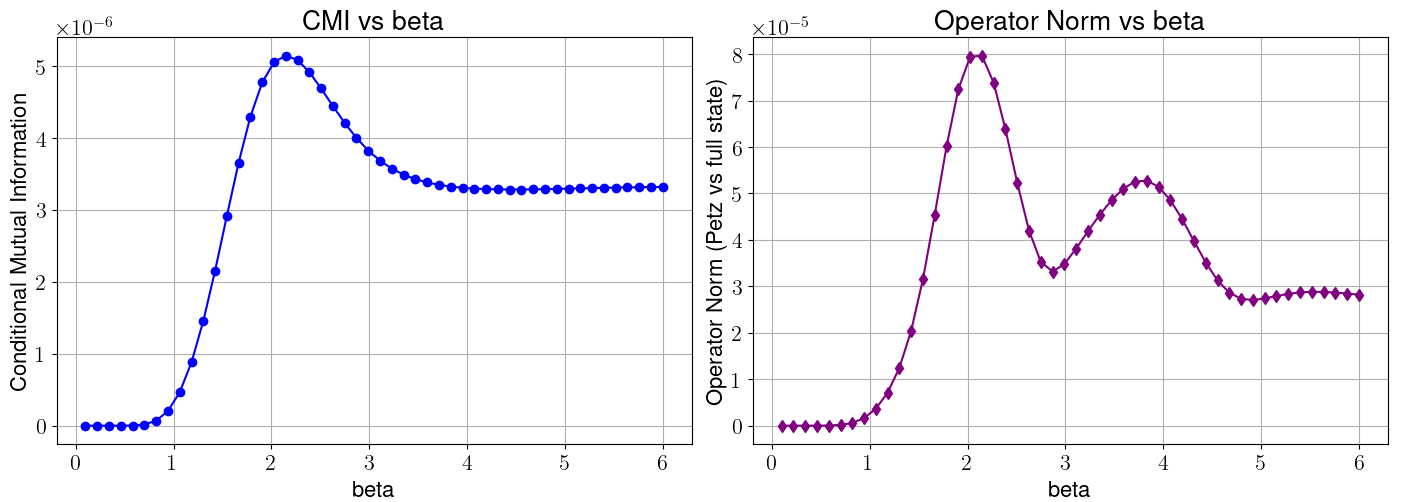}
    \caption{Reconstruction for thermal states where the sites are not permuted in the chaotic phase.}
    \label{fig:nonpermchaos}
\end{subfigure} %
\begin{subfigure}[t]{0.45\textwidth}
    \centering
    \includegraphics[width=\textwidth]{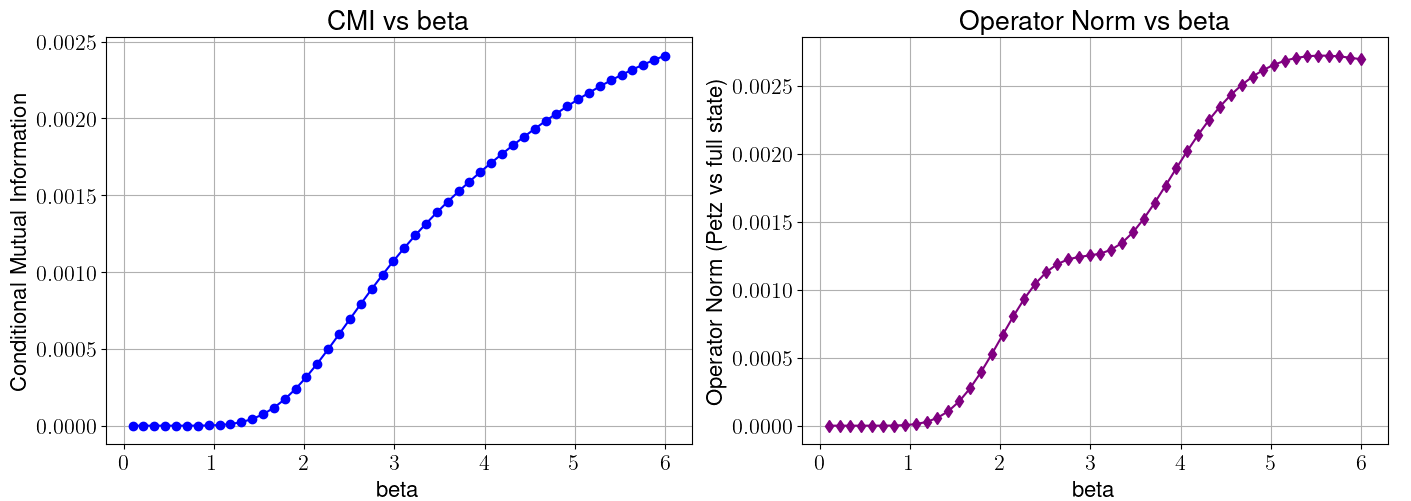}
    \caption{Reconstruction for thermal states where the sites are not permuted in the integrable phase.}
    \label{fig:nonperminteg}
\end{subfigure}
\begin{subfigure}[t]{0.45\textwidth}
    \centering
    \includegraphics[width=\textwidth]{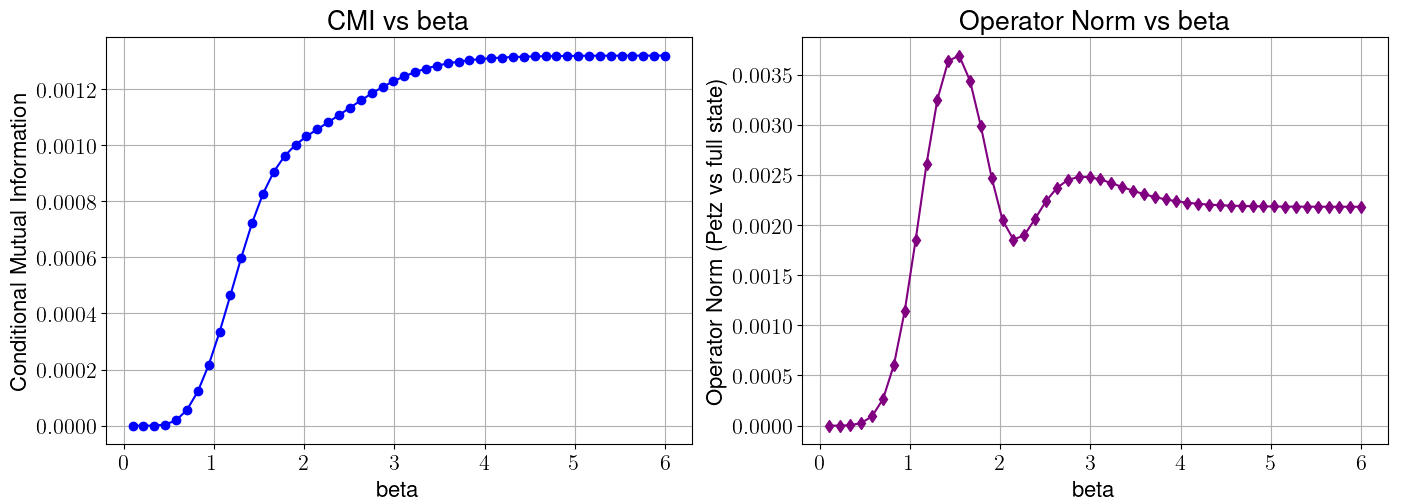}
    \caption{Reconstruction for thermal states where the sites are permuted in the chaotic phase.}
    \label{fig:permchaos}
\end{subfigure} %
\begin{subfigure}[t]{0.45\textwidth}
    \centering
    \includegraphics[width=\textwidth]{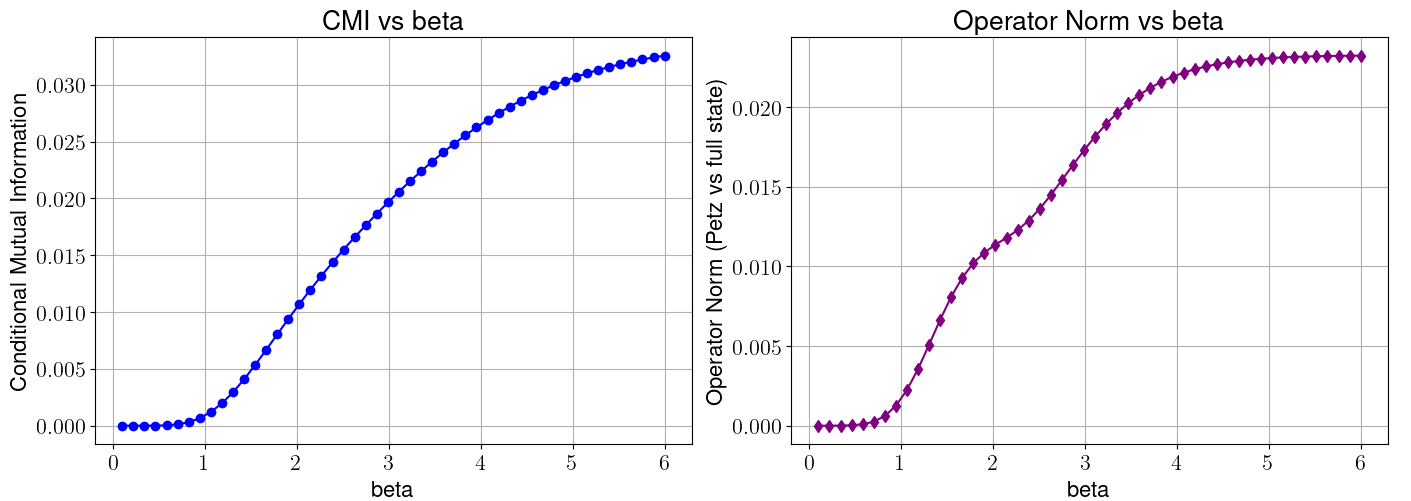}
    \caption{Reconstruction for thermal states where the sites are permuted in the integrable phase.}
    \label{fig:perminteg}
\end{subfigure}%
\caption{The CMI and Petz recovery operator norm in the integrable and chaotic phases}
\label{fig:2dplots}
\end{figure*}

The model that we study is given in Eq \eqref{eq:ham}. We begin by constructing the Gibbs state of the Hamiltonian 
\begin{equation}
    \rho_{\text{Gibbs}} = \frac{e^{- \beta H}}{Z}, \quad Z=\Tr(e^{-\beta H}).
\end{equation}
Note that the observed Wigner-Dyson statistics in Figure \ref{fig:WD} are observed within a charge sector for a specific parameter choice. When constructing the Gibbs state, we do so by constructing the thermal state in the full spin basis, and not the basis within a single charge sector. We will consider the Gibbs state for the Hamiltonian at different parameters, those corresponding to the chaotic phase and then we tune $h_z$ in the Hamiltonian to move towards the integrable phase. Furthermore, we do numerics on a spin chain with $8$ spins as Petz recovery becomes more numerically demanding otherwise and cannot be efficiently done on a laptop computer. However, we demonstrate in Figure \ref{fig:CMIinv} that for larger system sizes the universal features remain. 

\begin{figure}
    \centering
    \includegraphics[width=\linewidth]{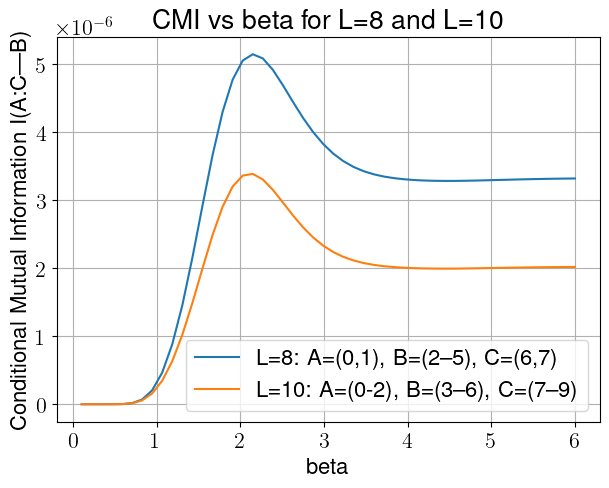}
    \caption{Invariance of CMI features for larger system sizes. }
    \label{fig:CMIinv}
\end{figure}

We divide the system into subsystems $A,B$ and $C$ in two different ways. In the first case we do not permute the indices and construct the thermal state where we identify the subsystems as $A = (0,1), B = (2,3,4,5), C = (6,7)$. In the second we permute them by applying a series of SWAP operators in the spin basis and then construct the Gibbs state. For notational convenience, we will describe this circuit of SWAP operators as the PERMUTE circuit. The particular instance of the PERMUTE circuit that we make use of on the $8$ site spin chain is as: 
\begin{equation}
\text{PERMUTE}: \quad [0,1,2,3,4,5,6,7] \xrightarrow{} [3,4,1,2,5,6,0,7]    
\end{equation} 
so that we can identify the subsystems as $A = (3,4), B = (1,2,5,6), C = (0,7)$ in the original lattice.

\begin{figure}[ht!]
    \centering
    \includegraphics[width=\linewidth]{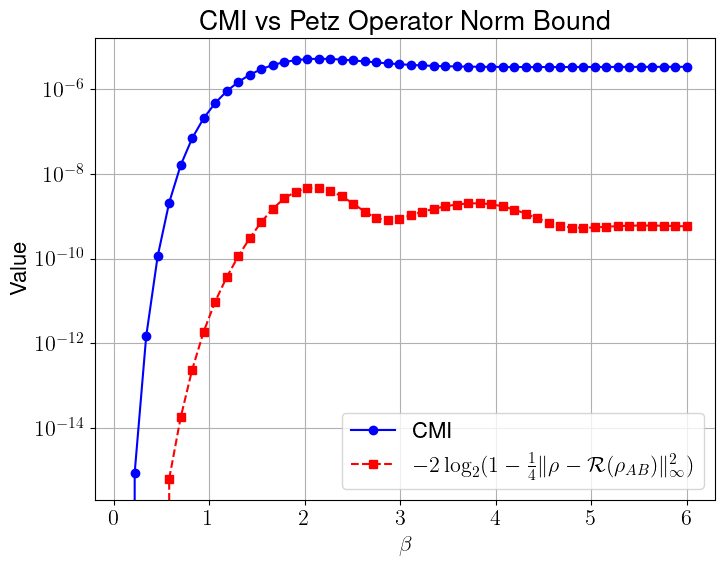}
    \caption{Numerical check of Equation \eqref{eq:figbound}.}
    \label{fig:bound}
\end{figure}

We visually depict the two configurations in Figure \ref{fig:configs}. For these two cases, after diagonalizing the Hamiltonian, constructing the eigenstates, choosing to apply PERMUTE on the eigenstates or not, we construct the Gibbs state which we will define as $\rho_{\text{Gibbs}}^{\text{non-permuted}}$ and $\rho_{\text{Gibbs}}^{\text{permuted}}$. The results are shown in Figure \ref{fig:2dplots}. We first computed the conditional mutual information for the two Gibbs state in both the integrable and chaotic phase of the Hamiltonian. We comment on a number of interesting features. It is important to remember that $\beta$ is inverse temperature.
From a physical perspective, low $\beta$ corresponds to a high temperature state while infinite $\beta$ corresponds to the ground state. The CMI is $0$ at $\beta=0$ because that's when the thermal state is essentially a product state. We will see that there is always a critical temperature after which we start to observe non-zero CMI.

A comment on why we used the operator norm instead of the fidelity - due to numerical difficulties, the operator norm was numerically more well behaved. By norm monotonicity $\|O\|_\infty \le \|O\|_1$ and the Fuchs--van de Graaf inequality
$D(\rho,\sigma)=\tfrac12\|\rho-\sigma\|_1 \le \sqrt{1-F(\rho,\sigma)}$, hence
$F(\rho,\sigma) \le 1 - \tfrac14\|\rho-\sigma\|_\infty^2$.
Combining with the bound
$F(\rho_{ABC},\widetilde{\rho}_{ABC}) \ge 2^{-I(A:C|B)/2}$ yields
\begin{equation}
I(A:C|B) \ge -2\log_2\!\Big(1-\tfrac14\|\rho-\mathcal{R}(\rho_{AB})\|_\infty^2\Big).\label{eq:figbound}
\end{equation}

We verify this bound numerically in the context of our computed quantities in Figure \ref{fig:bound}. We see the bound works well. The bound is not tight for higher $\beta$ - both quantities go to zero at infinite temperature. However, they display similar universal qualitative behavior.

Let's first consider the integrable case. In both the permuted and the non-permuted cases, the CMI is zero up until a critical temperature, after which it starts to increase and then plateaus. This can correspondingly be seen in the Petz recovery closeness between the recovered state and the original state. 

\begin{figure*}[h!]
\centering
\begin{subfigure}[b]{0.45\textwidth}
    \centering
    \includegraphics[width=\textwidth]{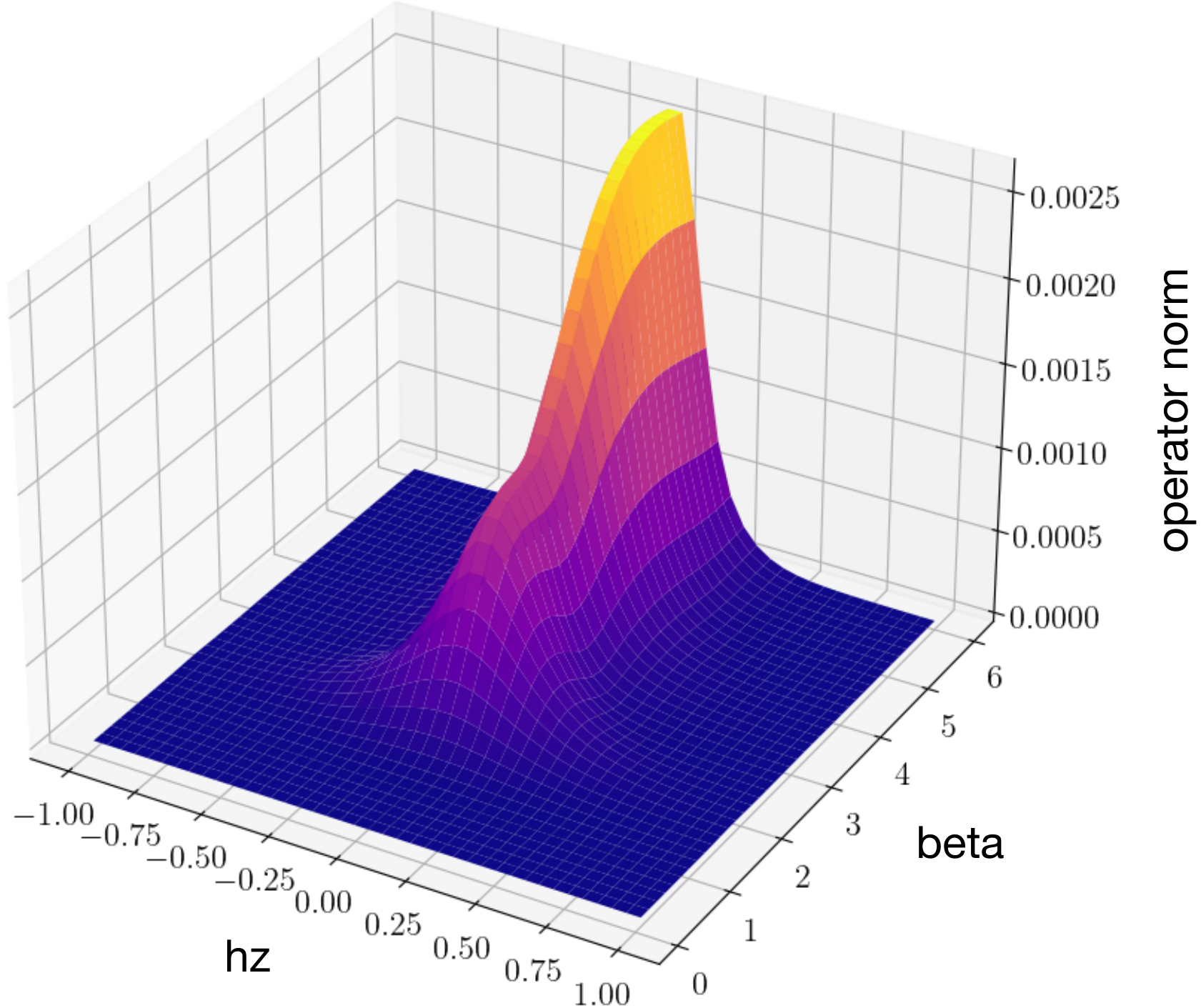}
    \caption{Operator norm difference, non-permuted}
    \label{fig:nonpermcpnom}
\end{subfigure} %
\begin{subfigure}[b]{0.45\textwidth}
    \centering
    \includegraphics[width=\textwidth]{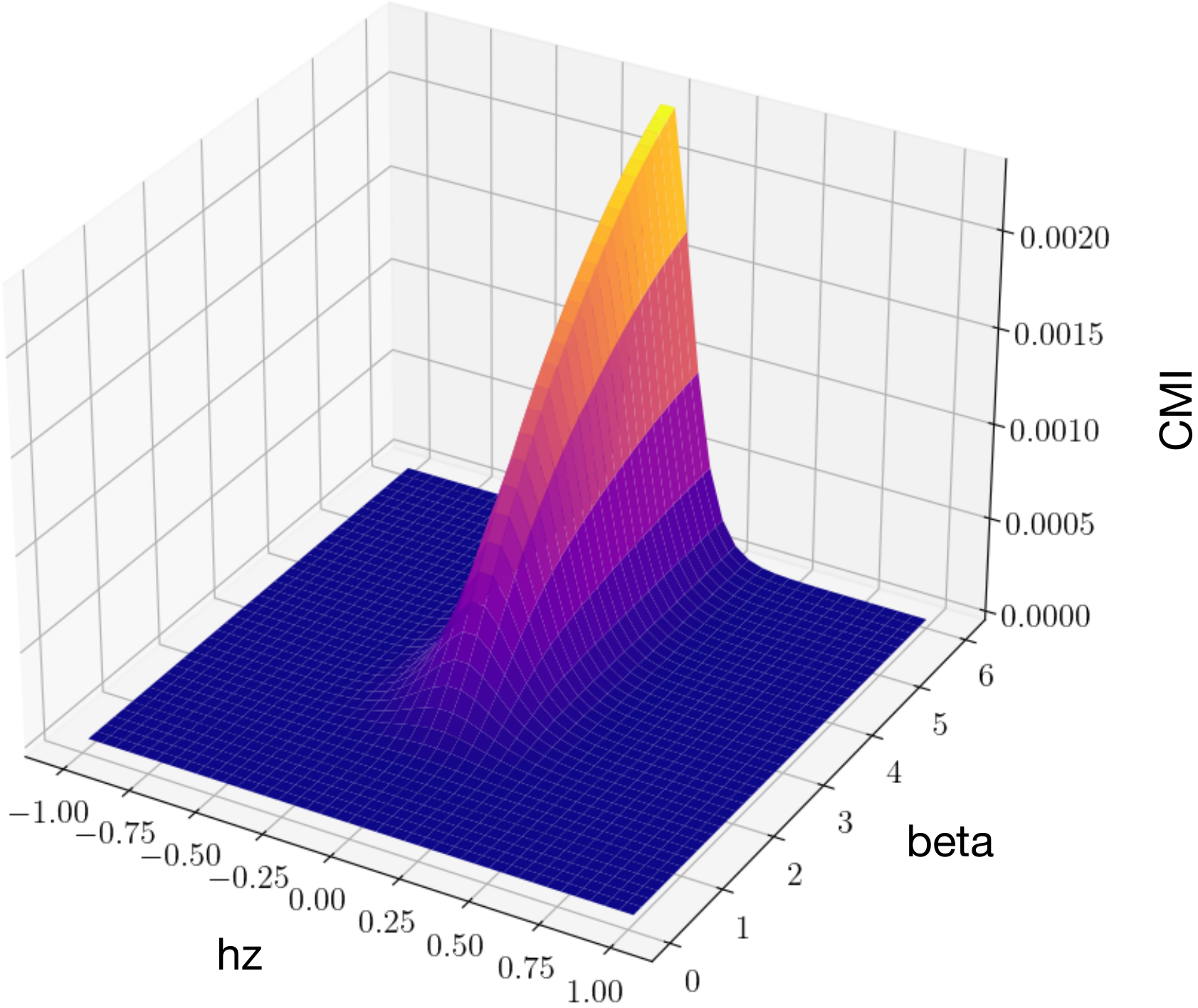}
    \caption{Conditional mutual information, non-permuted}
    \label{fig:nonpermcmi}
\end{subfigure}
\begin{subfigure}[b]{0.45\textwidth}
    \centering
    \includegraphics[width=\textwidth]{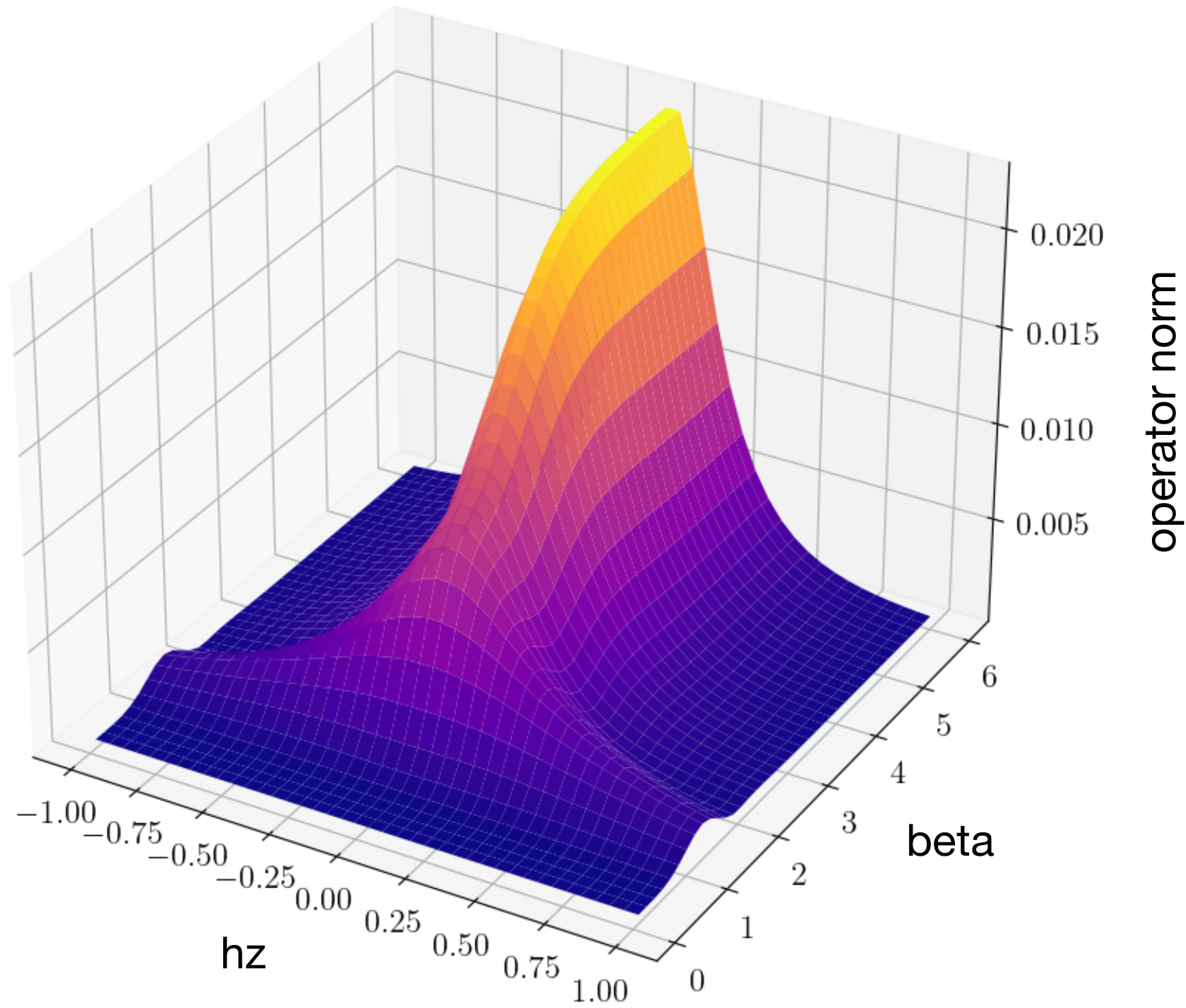}
    \caption{Operator norm difference, permuted}
    \label{fig:permopnom}
\end{subfigure} %
\begin{subfigure}[b]{0.45\textwidth}
    \centering
    \includegraphics[width=\textwidth]{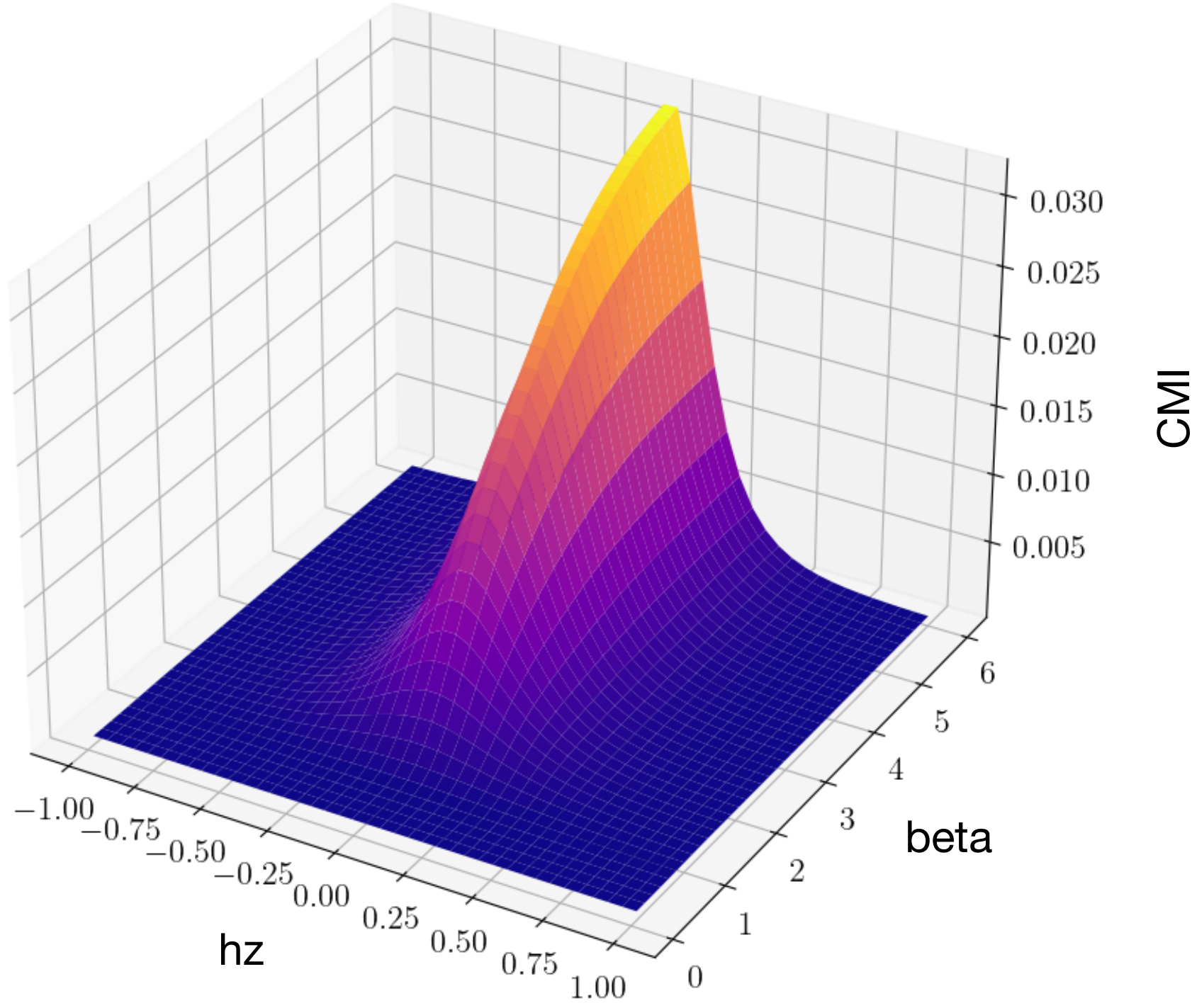}
    \caption{Conditional mutual information, permuted}
    \label{fig:permcmi}
\end{subfigure}%
\caption{Dependence of reconstruction on $\beta$ and $h_z$ for both permuted and non-permuted cases. The plots in Figure \ref{fig:2dplots} may be obtained as vertical slices (fixed $h_z$ of the 3d plots here.}
\label{fig:3dplots}
\end{figure*}

By contrast, compare the chaotic case. The CMI is much lower in the non-permuted case compared to the permuted case. Reconstruction is thus better in the non-permuted case. The exact form of the CMI structurally appears a bit different at intermediate temperatures, but at infinite and zero temperatures, we observe a saturation. Consequently, the reconstruction as well, and we observe an extra peak in the non-permuted case. The important point to note is that in the chaotic phase, the reconstruction is worse at some finite critical temperature - above and below that temperature, it is easier to reconstruct. 

Let's now compare the integrable and chaotic phases. The first observation is that at finite temperature, Petz recovery is better in the chaotic phase than it is in the integrable phase. The recovery for both is perfect at infinite temperature ($\beta=0$). Figure \ref{fig:3dplots} displays how well the state is recovered as $h_z$ is varied in the Hamiltonian. Recovery becomes harder as temperature is decreased but it always increases in the integrable case until it saturates - in the chaotic case it increases, peaks, decreases and then saturates.

\subsection{Thermal reconstruction with perturbative random band matrices}
\label{thermalrecons}

In the previous subsection, we did numerical reconstruction in a spin chain model and checked how well the process succeed. We would like to understand reconstruction analytically as well. To do so, we will work with the following random matrix inspired model of a spin chain Hamiltonian
\begin{equation}
    H_{mn} = O \delta_{mn} + DR_{mn},
\end{equation}
where $O$ and $D$ are constants and $R_{mn}$ are real Gaussian variables with mean zero and variance one, $R_{mn}$ is non-zero for $|m-n| \leq 1$, in other words it is a tridiagonal matrix. One can interpret the second term as a perturbation of a random hopping Hamiltonian.
The thermal density matrix $\rho= \frac{e^{-\beta H}}{Z}$ with the normalization constant $Z=\Tr(e^{-\beta H})$ can be expressed as 
\begin{equation}
    \rho = \frac{1}{Z}e^{- \beta O}\mathbb{I} \times e^{-\beta D R}.
\end{equation}
We will now analyze the system perturbatively i.e. for small $D$. As the computations are lengthy we quote the results here and defer the details to the Appendix. We find that the conditional mutual information for the random band matrix is given by 
\begin{align}
    I&(A:C|B)_{RB} = \frac{\beta^2 D^2}{2d^2} d_B \Tr(R_B^2) + \frac{\beta^2 D^2}{2d^2}d_B d \Tr R^2 \nonumber \\
    & - \frac{\beta^2 D^2}{2d^2} d_{AB}\Tr(R_{AB}^2) - \frac{\beta^2 D^2}{2d^2} d_{BC}\Tr(R_{BC}^2).
\end{align}

We now proceed to find the fidelity between the thermal and reconstructed states by computing the Petz map at leading order. In this limit, 
and we find that the fidelity is lower bounded by 
\begin{equation}
   2^{-\frac{1}{2}I(A:C|B)_{RB}}  \leq F(\rho_{ABC}, \mathcal{R}_B(\rho_{BC})).
\end{equation}

\section{Discussion}\label{sec:discussion}
We have established that in a particular case, when the Markovianity assumption is satisfied in chaotic quantum many-body systems, then one can successfully reconstruct the global state from its marginals. Our work can be considered as a procedure that given the promise of Markovianity, the likelihood of efficient reconstruction is guaranteed. 

Our work holds potential promise for practical tomography: an experimentalist could envisage restricting measurements to subsystems and then from subsystem data, reconstruct a global thermal state. It would be interesting to connect this to current efforts in tomography \cite{Gross_2010, Flammia_2012, Eisert_2020, Huang_2020, Huang_2021}. 

Another direction which we hope to return to in the future is an analysis of the output state of a random quantum circuit with mid-circuit measurements - see \cite{Potter_2022, Fisher_2023, skinner2023lecturenotesintroductionrandom} and references therein. 

In this work we analyzed the recovery of the temperature dependence of recovery - it would be interesting to conduct an investigation on the dependence of recovery on measurement strengths. Perhaps it might also be possible to make a connection of identifying Hamiltonians from single eigenstates \cite{Qi_2019,Garrison_2018}.

One might envisage applying our techniques to quantum chaotic models like the SYK \cite{Maldacena_2016,kitaev2015talks} The immediate difficulty is that these models in general have no meaningful notion of locality or subsystems - additionally, the Markovianity assumption was extremely important for the functioning of the Petz recovery map. Perhaps SYK chains and other inspired models \cite{Chowdhury_2022} may hold promise in this respect and be amenable to tomographic analysis by means of the Petz recovery techniques we have explored in this paper. 

Another interesting direction is to consider the possible impact on reconstruction of thermal CFT states \cite{Rodriguez_Gomez_2021, Datta_2019}. There has been some work on Petz recovery for ground states of lattice systems in terms of CFT \cite{vardhan2023Petzrecoverysubsystemsconformal} and various usages of the Petz map in the context of holography \cite{Cotler_2019, Chen_2020, penington2020replicawormholesblackhole}. However accessing properties of thermal states is much more difficult. Perhaps it may be possible to combine our techniques with either the conformal bootstrap  \cite{Poland_2019} or operator algebraic approaches \cite{Buchholz_2006} to investigate this problem.  It is possible to recast the marginal problem as an SDP - however, one gets guarantees using the Petz map. The SDP problem is to find $\sigma_{XYZ}$ subject to: $tr_Z(\sigma_{XYZ})=\rho_{XY}$, $tr_X(\sigma_{XYZ})=\rho_{YZ}$, $tr_Y(\sigma_{XYZ})=\rho_{XZ}$, $\sigma_{XYZ}\geq 0, \quad tr(\sigma_{XYZ})=1$. This may help with bootstrap investigations of the problem. With respect to holography, it is possible that the critical temperature we have identified after which recovery becomes easier at higher temperatures might be related to connected/disconnect entanglement wedge transitions in holography or page curve and other transitions in black holes \cite{PhysRevLett.44.301,stephens2001notesblackholephase}. 

Another interesting direction is to make modifications of the Petz recovery map. Currently our reconstruction works by looking at the full thermal state. However, Wigner-Dyson statistics is observed only within a single superselection sector. It would be interesting to see if one could reconstruct global states from subsystem information within superselection sectors. There are significant difficulties with this as it is not clear how to define a local subsystem within a superselection sector generally. However, in the special case that the superselection sectors corresponds to eigenspaces of a unitary rearrangement of the tensor product structure, one can construct the Gibbs state within that super-selection sector.\footnote{We would like to thank David Gross for pointing this out to us.} For example, in the case of the spin chain model we considered, the parity symmetry decomposes the Hilbert space into wavefunctions that are symmetric and anti-symmetric under reflection. One can construct a Gibbs state using the symmetric subsector and then use our permuted configuration to perform the decomposition into subsystems. The decompositions of the symmetric wavefunction within the permuted configuration will respect the tensor product structure in real space, thus allowing us to meaningully talk about subsystems in a non-local basis in real space. More generally such a program is ambitious and will presumably require the incorporation of new formal techniques like the entanglement bootstrap \cite{Shi_2021, Shi_2020}.

Finally, our work suggests that there may be numerical algorithms which could be constructed for a subclass of chaotic many-body systems, for which thermal physical quantities may be computed solely from the marginals - such efforts could assist in simplifying numerical quantum many-body physics ranging from high-energy physics (holographic CFTs and particle physics) to low-energy physics (quantum chemistry and condensed matter).


\subsection*{Acknowledgments}
We would like to thank Ari Chakraborty, Vatsal Dwivedi, David Gross, Alex Maloney, Philipp Schmoll and Greg White for extremely useful conversations and discussions. 
SQ acknowledges funding from the Einstein
Research Unit on Quantum Devices, Berlin Quantum, the
DFG (CRC 183, FOR 2724), and the
European Research Council (DebuQC).

\appendix

We analytically compute the CMI and bound the fidelity recovery via the Petz map for the spin chain Hamiltonian
\begin{equation}
    H_{mn} = O \delta_{mn} + DR_{mn},
\end{equation}
where $O$ and $D$ are constants and $R_{mn}$ are real Gaussian variables with mean zero and variance one, $R_{mn}$ is non-zero for $|m-n| \leq 1$. 
The thermal density matrix $\rho= \frac{e^{-\beta H}}{Z}$ with the normalization constant $Z=\Tr(e^{-\beta H})$ can be expressed as 
\begin{equation}
    \rho = \frac{1}{Z}e^{- \beta O}\mathbb{I} \times e^{-\beta D R}.
\end{equation}
We will now analyze the system perturbatively i.e. for small $D$ in which case the thermal density matrix can be written as 
\begin{equation}
    \rho = \frac{1}{Z}e^{-\beta O} \left( 1 - \beta D R + \frac{(\beta D R)^2}{2} + \dots \right).
\end{equation}
For $D=0$, we obtain the maximally mixed state 
\begin{equation}
    \rho_0 = \frac{e^{-\beta O \mathbb{I}}}{\Tr(e^{-\beta O \mathbb{I}})} = \frac{e^{-\beta O}\mathbb{I}}{e^{-\beta O}\Tr(\mathbb{I})} = \frac{\mathbb{I}}{d}.
\end{equation}
For the maximally mixed state, the CMI will be zero and we will obtain perfect Petz recovery. We can tripartition the system into three parts $\rho_{ABC}= \frac{\mathbb{I}_{ABC}}{d}$ where the dimensions are given by $d_i = \dim \mathcal{H}_i$ for $i \in \{A,B,C\}$ where $d=d_A d_B d_C$ - the reduced density matrices are given by $\rho_i = \frac{\mathbb{I}_i}{d_i}$. The entropies are respectively given by $S(R_i)=\log(d_i)$ for the region $R$, for example $S(AB)=\log(d_Ad_B)$ while the total entropy $S(ABC)=\log(d)$, a short computation shows that 
\begin{equation}
    I(A:C|B) = 0,
\end{equation}
for the maximally mixed state. 

We will now expand the Von Neumann entropy for a perturbation of the form 
\begin{equation}
\rho(D) = \rho_0 + D\rho_1 + \mathcal{O}(D^2),
\end{equation}
where $\rho_0$ is the maximally mixed state and $\rho_1 = -\frac{\beta}{d}(R-\frac{\Tr(R)}{d}1)$. 

We now use the fact that for a density matrix $\rho=\frac{I}{d}+\epsilon \rho_1$ where $\rho_1=\rho_1^\dagger$ and $\Tr(\rho_1)=0$, then for $\norm{\epsilon d \rho_1} <1$, $S(\rho)=\log d - \frac{\epsilon^2}{2}d\Tr(\rho_1^2) + O(\epsilon^3)$. This can be shown using the starting point $\rho=\frac{I}{d}+ X$ where $X:=\epsilon d \rho_1$. We can expand the logarithm $\log(1+X)=X-\frac{1}{2}X^2 + O(\norm{X}^3)$ where $\norm{X}<1$. We have that 
\begin{align}
    \log(\rho) = \log \left(\frac{1}{d}(I+X)\right)&=-\log d I + X - \frac{1}{2}X^2 \nonumber \\
    &+ O(\norm{X}^3),
\end{align}
where from now on, we'll drop higher order terms. From this we obtain the vVon Neumann entropy as 
\begin{align}
    S(\rho) &= -\Tr \left[\frac{1}{d}(I+X)(-\log d + X - \frac{1}{2}X^2) \right], \\
    &= \log d - \frac{\epsilon^2}{2}d\Tr(\rho_1^2).
\end{align}
We now take the dimensions of the random band matrix to be $d_A, d_B, d_C$ and $d=d_A d_B d_C$. For $\rho=\frac{I}{d} + D \rho_1$,
\begin{equation}
    S(\rho_X)=\log d_x - \frac{D^2}{2}d_x \Tr(\rho_{1,X}^2),
\end{equation}
where $X \subseteq \{A,B,C\}$ and $\rho_{1,X}=\Tr_{X^c} \rho_1$. The CMI in this case reduces to 
\begin{align}
    I(A:C|B) = & \frac{D^2}{2}  d_B \Tr(\rho_{1,B}^2) + \frac{D^2}{2}d \Tr(\rho_{1,ABC}^2) \nonumber \\
    &- \frac{D^2}{2} d_{AB}\Tr(\rho_{1,AB}^2) - \frac{D^2}{2} d_B \Tr(\rho_{1,BC}^2).
\end{align}
Given our Hamiltonian $H=OI + DR$, and $\rho_1= \frac{-\beta}{d}(R-\frac{\Tr (R)}{d})$. So for any subset $X$, 
\begin{equation}
    \Tr(\rho_{1,X}^2)=\frac{\beta^2}{d^2}\left(\Tr(R_X^2)-\frac{(\Tr R_X)^2}{d_X} \right),
\end{equation}
where $R_X :=\Tr_{X^c} R$ and $\Tr R_X = \Tr R$ for averages. Which gives the CMI  
\begin{align}
    I&(A:C|B)_{RB} = \frac{\beta^2 D^2}{2d^2} d_B \Tr(R_B^2) + \frac{\beta^2 D^2}{2d^2}d_B d \Tr R^2 \nonumber \\
    & - \frac{\beta^2 D^2}{2d^2} d_{AB}\Tr(R_{AB}^2) - \frac{\beta^2 D^2}{2d^2} d_{BC}\Tr(R_{BC}^2).
\end{align}
This immediately gives us a lower bound on the fidelity between the recovered state and the original state. 

\bibliographystyle{unsrt}
\bibliography{paper}{}

\begin{thebibliography}{10}

\bibitem{Coleman}
A.~J. Coleman.
\newblock Structure of fermion density matrices.
\newblock {\em Rev. Mod. Phys.}, 35:668--686, Jul 1963.

\bibitem{marginalphysical}
{Schilling, Christian}.
\newblock {\em Quantum marginal problem and its physical relevance}.
\newblock PhD thesis, ETH Zurich, 2014.

\bibitem{Klyachko_2006}
Alexander~A Klyachko.
\newblock Quantum marginal problem and n-representability.
\newblock {\em Journal of Physics: Conference Series}, 36:72–86, April 2006.

\bibitem{Mazziotti_2012}
David~A. Mazziotti.
\newblock Structure of fermionic density matrices: Complete n representability conditions.
\newblock {\em Physical Review Letters}, 108(26), June 2012.

\bibitem{Mazziotti_2023}
David~A. Mazziotti.
\newblock Quantum many-body theory from a solution of the n -representability problem.
\newblock {\em Physical Review Letters}, 130(15), April 2023.

\bibitem{klyachko2004quantummarginalproblemrepresentations}
Alexander Klyachko.
\newblock Quantum marginal problem and representations of the symmetric group, 2004.

\bibitem{Lopes2015}
Alexandre Miguel de~Araújo Lopes.
\newblock {\em Pure univariate quantum marginals and electronic transport properties of geometrically frustrated systems}.
\newblock Dissertation, Universität Freiburg, 2015.

\bibitem{eisert2023notelowerboundsvariational}
J.~Eisert.
\newblock A note on lower bounds to variational problems with guarantees, 2023.

\bibitem{Eisert_2008}
Jens Eisert, Tomáš Tyc, Terry Rudolph, and Barry~C. Sanders.
\newblock Gaussian quantum marginal problem.
\newblock {\em Communications in Mathematical Physics}, 280(1):263–280, February 2008.

\bibitem{Walter_2013}
Michael Walter, Brent Doran, David Gross, and Matthias Christandl.
\newblock Entanglement polytopes: Multiparticle entanglement from single-particle information.
\newblock {\em Science}, 340(6137):1205–1208, June 2013.

\bibitem{trimborn}
F~Trimborn, Reinhard Werner, and Dirk Witthaut.
\newblock Quantum de finetti theorems and mean-field theory from quantum phase space representations.
\newblock {\em Journal of Physics A: Mathematical and Theoretical}, 49:135302, 02 2016.

\bibitem{TNreview}
J.~Ignacio Cirac, David P\'erez-Garc\'{\i}a, Norbert Schuch, and Frank Verstraete.
\newblock Matrix product states and projected entangled pair states: Concepts, symmetries, theorems.
\newblock {\em Rev. Mod. Phys.}, 93:045003, Dec 2021.

\bibitem{Ba_uls_2023}
Mari~Carmen Bañuls.
\newblock Tensor network algorithms: A route map.
\newblock {\em Annual Review of Condensed Matter Physics}, 14(1):173–191, March 2023.

\bibitem{kim1}
Isaac~H. Kim.
\newblock Entropy scaling law and the quantum marginal problem.
\newblock {\em Phys. Rev. X}, 11:021039, May 2021.

\bibitem{kim2}
Isaac~H. Kim.
\newblock Entropy scaling law and the quantum marginal problem: simplification and generalization, 2021.

\bibitem{thermalization}
Dmitry~A. Abanin, Ehud Altman, Immanuel Bloch, and Maksym Serbyn.
\newblock Colloquium: Many-body localization, thermalization, and entanglement.
\newblock {\em Rev. Mod. Phys.}, 91:021001, May 2019.

\bibitem{Srednicki:1994mfb}
Mark Srednicki.
\newblock {Chaos and Quantum Thermalization}.
\newblock {\em Phys. Rev. E}, 50, 3 1994.

\bibitem{Deutsch:1991msp}
J.~M. Deutsch.
\newblock {Quantum statistical mechanics in a closed system}.
\newblock {\em Phys. Rev. A}, 43(4):2046, 1991.

\bibitem{Jensen:1984gu}
R.~V. Jensen and R.~Shankar.
\newblock {Statistical Behavior in Deterministic Quantum Systems With Few Degrees of Freedom}.
\newblock {\em Phys. Rev. Lett.}, 54:1879, 1985.

\bibitem{Hu_2024}
Yangrui Hu and Yijian Zou.
\newblock Petz map recovery for long-range entangled quantum many-body states.
\newblock {\em Physical Review B}, 110(19), November 2024.

\bibitem{vardhan2023Petzrecoverysubsystemsconformal}
Shreya Vardhan, Annie~Y. Wei, and Yijian Zou.
\newblock Petz recovery from subsystems in conformal field theory, 2023.

\bibitem{jia2020Petzreconstructionrandomtensor}
Hewei~Frederic Jia and Mukund Rangamani.
\newblock Petz reconstruction in random tensor networks, 2020.

\bibitem{temme2015faststabilizerhamiltoniansthermalize}
Kristan Temme and Michael~J. Kastoryano.
\newblock How fast do stabilizer hamiltonians thermalize?, 2015.

\bibitem{Alhambra_2017}
Álvaro~M. Alhambra and Mischa~P. Woods.
\newblock Dynamical maps, quantum detailed balance, and the petz recovery map.
\newblock {\em Physical Review A}, 96(2), August 2017.

\bibitem{BGS1984}
O.~Bohigas, M.~J. Giannoni, and C.~Schmit.
\newblock Characterization of chaotic quantum spectra and universality of level fluctuation laws.
\newblock {\em Physical Review Letters}, 52(1):1--4, 1984.

\bibitem{qasim2025emergentstatisticalmechanicsholographic}
Shozab Qasim, Jens Eisert, and Alexander Jahn.
\newblock Emergent statistical mechanics in holographic random tensor networks, 2025.

\bibitem{Haferkamp_2021}
Jonas Haferkamp, Christian Bertoni, Ingo Roth, and Jens Eisert.
\newblock Emergent statistical mechanics from properties of disordered random matrix product states.
\newblock {\em PRX Quantum}, 2(4), October 2021.

\bibitem{Harrow:2022znr}
Aram~W. Harrow and Yichen Huang.
\newblock {Thermalization without eigenstate thermalization}.
\newblock 9 2022.

\bibitem{Bertoni_2025}
Christian Bertoni, Clara Wassner, Giacomo Guarnieri, and Jens Eisert.
\newblock Typical thermalization of low-entanglement states.
\newblock {\em Communications Physics}, 8(1), July 2025.

\bibitem{wang2025eigenstatethermalizationhypothesisrandom}
Jiaozi Wang, Hua Yan, Robin Steinigeweg, and Jochen Gemmer.
\newblock Eigenstate thermalization hypothesis and random matrix theory universality in few-body systems, 2025.

\bibitem{jafferis2023jtgravitymattergeneralized}
Daniel~Louis Jafferis, David~K. Kolchmeyer, Baur Mukhametzhanov, and Julian Sonner.
\newblock Jt gravity with matter, generalized eth, and random matrices, 2023.

\bibitem{jafferis2023matrixmodelseigenstatethermalization}
Daniel~Louis Jafferis, David~K. Kolchmeyer, Baur Mukhametzhanov, and Julian Sonner.
\newblock Matrix models for eigenstate thermalization, 2023.

\bibitem{cáceres2024genericetheigenstatethermalization}
Elena Cáceres, Stefan Eccles, Jason Pollack, and Sarah Racz.
\newblock Generic eth: Eigenstate thermalization beyond the microcanonical, 2024.

\bibitem{Abanin_2019}
Dmitry~A. Abanin, Ehud Altman, Immanuel Bloch, and Maksym Serbyn.
\newblock Colloquium : Many-body localization, thermalization, and entanglement.
\newblock {\em Reviews of Modern Physics}, 91(2), May 2019.

\bibitem{Gogolin_2016}
Christian Gogolin and Jens Eisert.
\newblock Equilibration, thermalisation, and the emergence of statistical mechanics in closed quantum systems.
\newblock {\em Reports on Progress in Physics}, 79(5):056001, April 2016.

\bibitem{Eisert_2015}
J.~Eisert, M.~Friesdorf, and C.~Gogolin.
\newblock Quantum many-body systems out of equilibrium.
\newblock {\em Nature Physics}, 11(2):124–130, February 2015.

\bibitem{schilling2014quantummarginalproblem}
Christian Schilling.
\newblock The quantum marginal problem, 2014.

\bibitem{Fawzi_2015}
Omar Fawzi and Renato Renner.
\newblock Quantum conditional mutual information and approximate markov chains.
\newblock {\em Communications in Mathematical Physics}, 340(2):575–611, September 2015.

\bibitem{10.21468/SciPostPhys.2.1.003}
Phillip Weinberg and Marin Bukov.
\newblock {QuSpin: a Python package for dynamics and exact diagonalisation of quantum many body systems part I: spin chains}.
\newblock {\em SciPost Phys.}, 2:003, 2017.

\bibitem{Gross_2010}
David Gross, Yi-Kai Liu, Steven~T. Flammia, Stephen Becker, and Jens Eisert.
\newblock Quantum state tomography via compressed sensing.
\newblock {\em Physical Review Letters}, 105(15), October 2010.

\bibitem{Flammia_2012}
Steven~T Flammia, David Gross, Yi-Kai Liu, and Jens Eisert.
\newblock Quantum tomography via compressed sensing: error bounds, sample complexity and efficient estimators.
\newblock {\em New Journal of Physics}, 14(9):095022, September 2012.

\bibitem{Eisert_2020}
Jens Eisert, Dominik Hangleiter, Nathan Walk, Ingo Roth, Damian Markham, Rhea Parekh, Ulysse Chabaud, and Elham Kashefi.
\newblock Quantum certification and benchmarking.
\newblock {\em Nature Reviews Physics}, 2(7):382–390, June 2020.

\bibitem{Huang_2020}
Hsin-Yuan Huang, Richard Kueng, and John Preskill.
\newblock Predicting many properties of a quantum system from very few measurements.
\newblock {\em Nature Physics}, 16(10):1050–1057, June 2020.

\bibitem{Huang_2021}
Hsin-Yuan Huang, Michael Broughton, Masoud Mohseni, Ryan Babbush, Sergio Boixo, Hartmut Neven, and Jarrod~R. McClean.
\newblock Power of data in quantum machine learning.
\newblock {\em Nature Communications}, 12(1), May 2021.

\bibitem{Potter_2022}
Andrew~C. Potter and Romain Vasseur.
\newblock {\em Entanglement Dynamics in Hybrid Quantum Circuits}, page 211–249.
\newblock Springer International Publishing, 2022.

\bibitem{Fisher_2023}
Matthew~P.A. Fisher, Vedika Khemani, Adam Nahum, and Sagar Vijay.
\newblock Random quantum circuits.
\newblock {\em Annual Review of Condensed Matter Physics}, 14(1):335–379, March 2023.

\bibitem{skinner2023lecturenotesintroductionrandom}
Brian Skinner.
\newblock Lecture notes: Introduction to random unitary circuits and the measurement-induced entanglement phase transition, 2023.

\bibitem{Qi_2019}
Xiao-Liang Qi and Daniel Ranard.
\newblock Determining a local hamiltonian from a single eigenstate.
\newblock {\em Quantum}, 3:159, July 2019.

\bibitem{Garrison_2018}
James~R. Garrison and Tarun Grover.
\newblock Does a single eigenstate encode the full hamiltonian?
\newblock {\em Physical Review X}, 8(2), April 2018.

\bibitem{Maldacena_2016}
Juan Maldacena, Stephen~H. Shenker, and Douglas Stanford.
\newblock A bound on chaos.
\newblock {\em Journal of High Energy Physics}, 2016(8), August 2016.

\bibitem{kitaev2015talks}
Alexei Kitaev.
\newblock A simple model of quantum holography.
\newblock \href{http://online.kitp.ucsb.edu/online/entangled15/kitaev/}{Talk at KITP, April 7, 2015}; \href{http://online.kitp.ucsb.edu/online/entangled15/kitaev2/}{Talk at KITP, May 27, 2015}, 2015.
\newblock Talks at the Kavli Institute for Theoretical Physics.

\bibitem{Chowdhury_2022}
Debanjan Chowdhury, Antoine Georges, Olivier Parcollet, and Subir Sachdev.
\newblock Sachdev-ye-kitaev models and beyond: Window into non-fermi liquids.
\newblock {\em Reviews of Modern Physics}, 94(3), September 2022.

\bibitem{Rodriguez_Gomez_2021}
D.~Rodriguez-Gomez and J.~G. Russo.
\newblock Thermal correlation functions in cft and factorization.
\newblock {\em Journal of High Energy Physics}, 2021(11), November 2021.

\bibitem{Datta_2019}
Shouvik Datta, Per Kraus, and Ben Michel.
\newblock Typicality and thermality in 2d cft.
\newblock {\em Journal of High Energy Physics}, 2019(7), July 2019.

\bibitem{Cotler_2019}
Jordan Cotler, Patrick Hayden, Geoffrey Penington, Grant Salton, Brian Swingle, and Michael Walter.
\newblock Entanglement wedge reconstruction via universal recovery channels.
\newblock {\em Physical Review X}, 9(3), July 2019.

\bibitem{Chen_2020}
Chi-Fang Chen, Geoffrey Penington, and Grant Salton.
\newblock Entanglement wedge reconstruction using the petz map.
\newblock {\em Journal of High Energy Physics}, 2020(1), January 2020.

\bibitem{penington2020replicawormholesblackhole}
Geoff Penington, Stephen~H. Shenker, Douglas Stanford, and Zhenbin Yang.
\newblock Replica wormholes and the black hole interior, 2020.

\bibitem{Poland_2019}
David Poland, Slava Rychkov, and Alessandro Vichi.
\newblock The conformal bootstrap: Theory, numerical techniques, and applications.
\newblock {\em Reviews of Modern Physics}, 91(1), January 2019.

\bibitem{Buchholz_2006}
Detlev Buchholz, Claudio D’Antoni, and Roberto Longo.
\newblock Nuclearity and thermal states in conformal field theory.
\newblock {\em Communications in Mathematical Physics}, 270(1):267–293, November 2006.

\bibitem{PhysRevLett.44.301}
Don~N. Page.
\newblock Is black-hole evaporation predictable?
\newblock {\em Phys. Rev. Lett.}, 44:301--304, Feb 1980.

\bibitem{stephens2001notesblackholephase}
G.~J. Stephens and B.~L. Hu.
\newblock Notes on black hole phase transitions, 2001.

\bibitem{Shi_2021}
Bowen Shi and Isaac~H. Kim.
\newblock Entanglement bootstrap approach for gapped domain walls.
\newblock {\em Physical Review B}, 103(11), March 2021.

\bibitem{Shi_2020}
Bowen Shi, Kohtaro Kato, and Isaac~H. Kim.
\newblock Fusion rules from entanglement.
\newblock {\em Annals of Physics}, 418:168164, July 2020.

\end{thebibliography}

\end{document}